%% file: BYOO_AAMAS25.tex
\documentclass[sigconf, nonacm]{aamas}

\usepackage{balance} 
\usepackage[linesnumbered,ruled,vlined]{algorithm2e}
\usepackage{subfigure}
\usepackage{appendix}

\usepackage{ifthen}
\newboolean{include_appendix}

\setboolean{include_appendix}{true}

\theoremstyle{boldhead}
\newtheorem{assumption}{Assumption}
\newtheorem{theorem}{Theorem}
\newtheorem{definition}{Definition}
\newtheorem{corollary}{Corollary}
\newtheorem{lemma}{Lemma}

\makeatletter
\gdef\@copyrightpermission{
  \begin{minipage}{0.2\columnwidth}
   \href{https://creativecommons.org/licenses/by/4.0/}{\includegraphics[width=0.90\textwidth]{by}}
  \end{minipage}\hfill
  \begin{minipage}{0.8\columnwidth}
   \href{https://creativecommons.org/licenses/by/4.0/}{This work is licensed under a Creative Commons Attribution International 4.0 License.}
  \end{minipage}
  \vspace{5pt}
}
\makeatother

\setcopyright{ifaamas}
\acmConference[AAMAS '25]{Proc.\@ of the 24th International Conference
on Autonomous Agents and Multiagent Systems (AAMAS 2025)}{May 19 -- 3, 2025}
{Detroit, Michigan, USA}{Y.~Vorobeychik, S.~Das, A.~Nowé  (eds.)}
\copyrightyear{2025}
\acmYear{2025}
\acmDOI{}
\acmPrice{}
\acmISBN{}

\acmSubmissionID{<<OpenReview submission id>>}

\title[AAMAS-2025 Formatting Instructions]{Adaptive Episode Length Adjustment for Multi-agent Reinforcement Learning}

\author{Byunghyun Yoo}
\affiliation{
  \institution{Electronics and Telecommunications Research Institute (ETRI)}
  \city{Daejeon}
  \country{South Korea}}
\email{bhyoo@etri.re.kr}

\author{Younghwan Shin}
\affiliation{
  \institution{Electronics and Telecommunications Research Institute (ETRI)}
  \city{Daejeon}
  \country{South Korea}}
\email{shinyh1115@etri.re.kr}

\author{Hyunwoo Kim}
\affiliation{
  \institution{Electronics and Telecommunications Research Institute (ETRI)}
  \city{Daejeon}
  \country{South Korea}}
\email{kimhw@etri.re.kr}

\author{Euisok Chung}
\affiliation{
  \institution{Electronics and Telecommunications Research Institute (ETRI)}
  \city{Daejeon}
  \country{South Korea}}
\email{eschung@etri.re.kr}

\author{Jeongmin Yang}
\affiliation{
  \institution{Electronics and Telecommunications Research Institute (ETRI)}
  \city{Daejeon}
  \country{South Korea}}
\email{jmyang38@etri.re.kr}

\begin{abstract}
In standard reinforcement learning, an episode is defined as a sequence of interactions between agents and the environment, which terminates upon reaching a terminal state or a pre-defined episode length. 
Setting a shorter episode length enables the generation of multiple episodes with the same number of data samples, thereby facilitating an exploration of diverse states.
While shorter episodes may limit the collection of long-term interactions, they may offer significant advantages when properly managed. For example, trajectory truncation in single-agent reinforcement learning has shown how the benefits of shorter episodes can be leveraged despite the trade-off of reduced long-term interaction experiences. However, this approach remains underexplored in MARL.
This paper proposes a novel MARL approach, Adaptive Episode Length Adjustment (AELA), where the episode length is initially limited and gradually increased based on an entropy-based assessment of learning progress. 
By starting with shorter episodes, agents can focus on learning effective strategies for initial states and minimize time spent in dead-end states.
The use of entropy as an assessment metric prevents premature convergence to suboptimal policies and ensures balanced training over varying episode lengths.
We validate our approach using the StarCraft Multi-agent Challenge (SMAC) and a modified predator-prey environment, demonstrating significant improvements in both convergence speed and overall performance compared to existing methods. To the best of our knowledge, this is the first study to adaptively adjust episode length in MARL based on learning progress.
\end{abstract}


\newcommand{\BibTeX}{\rm B\kern-.05em{\sc i\kern-.025em b}\kern-.08em\TeX}

\begin{document}


\pagestyle{fancy}
\fancyhead{}


\maketitle 


\section{Introduction}

Multi-agent reinforcement learning (MARL) has become an increasingly important area of research, especially in scenarios involving cooperative or competitive environments with multiple autonomous agents~\cite{dinneweth2022multi, huttenrauch2017guided, ye2015multi}. The ability of agents to make decisions based on their individual perspectives while interacting with other agents presents a complex yet realistic setting that resembles many natural and artificial systems. As MARL evolves, researchers have identified several key challenges, such as non-stationarity, coordination among agents, and the curse of dimensionality. Among these challenges, the impact of episode length or time limits on the learning performance of agents is of particular interest, especially in practical settings where time constraints are inherent.

In MARL, an episode is defined as a sequence of interactions between agents and the environment, which terminates upon reaching a terminal state or a pre-defined episode length. Typically, the episode length is set to a fixed value throughout the training process, and this value is sufficiently large to allow for finding an optimal policy to solve the given ask~\cite{sutton2018reinforcement}. 
Determining an appropriate episode length is not straightforward. 
Setting a shorter episode length enables the generation of multiple episodes with the same number of data samples, thereby facilitating an exploration of diverse states. This increased exploration can help agents gather a wider variety of experiences, which is beneficial for generalization. 
However, it also has the drawback of potentially failing to collect experiences involving long-term interactions that exceed the restricted episode length.

Recent research in single-agent reinforcement learning has increasingly focused on techniques such as time limits or trajectory truncation, which utilize only a portion of an episode~\cite{pardo2018time, fuks2019evolution, poiani2023truncating, poiani2024truncating}. These methods have been demonstrated to be beneficial both theoretically and empirically. Adjusting episode length is expected to be even more effective in MARL, where the state-action space is significantly larger.
Especially in multi-agent environments, the complexity arising from interactions between agents can lead to the existence of various types of dead-end states where agents cannot reach a desired goal within a reasonable number of steps.
Dead-end states can severely hinder the learning process, especially if agents frequently encounter them in early training stages~\cite{fatemi2019dead, zhang2023safe, killian2023risk}. 
By adjusting episode length, agents can avoid spending excessive time in such unproductive states, focusing instead on learning effective strategies for early interactions. 
Despite the potential of these approaches, their application in MARL remains relatively underexplored.

In this paper, we propose a novel approach to MARL called Adaptive Episode Length Approach (AELA), where the episode length is initially limited during the early training stages and then gradually increased based on the learning situation.
AELA begins by restricting agents to shorter episodes, allowing them to focus on learning effective strategies for initial states and to reduce a prolonged time in dead-end states.
Then, our approach gradually increases the episode length as learning progresses, based on an entropy-based assessment of the agents. 
Entropy is used to prevent excessive convergence to experiences of a specific episode length, thereby ensuring adequate training over the originally defined episode length for the task.
Also, this staged increase in episode length allows agents to incrementally expand their exploration without being overwhelmed by the complexity of the full task from the outset.

We validated our proposed method through experiments conducted on the StarCraft Multi-agent Challenge (SMAC)\cite{samvelyan2019starcraft} and modified predator-prey (MPP)~\cite{son2019qtran}, which are widely used benchmarks for testing MARL algorithms. Our extensive evaluation reveals that our approach achieves significantly superior performance, particularly in terms of performance and convergence speed. The results demonstrate that limiting the episode length during the early training stages helps agents develop a stronger understanding of the initial state, ultimately improving final performance. To the best of our knowledge, this is the first study to adaptively adjust the episode length based on the learning situation in MARL.

\section{Related work}
\paragraph{Cooperative multi-agent reinforcement learning}

Methods that decompose joint action values have been extensively studied and have shown notable success in cooperative multi-agent reinforcement learning (MARL) scenarios. The Value Decomposition Network (VDN), introduced by ~\citet{sunehag2018value}, models the joint action-value function as the sum of individual action-value functions, which promotes cooperative behavior among agents. QMIX, proposed by~\citet{rashid2018qmix}, expanded on VDN by incorporating a mixing network that used monotonic utility functions under Individual-Global-Max (IGM) conditions, enhancing the flexibility of value decomposition while preserving the monotonicity constraint. To overcome the limitations of the monotonicity assumption, later approaches such as QTRAN~\cite{son2019qtran} and Weighted QMIX (WQMIX)~\cite{rashid2020weighted} were developed to relax these constraints by applying weights to the loss functions of suboptimal actions, thereby improving performance. Furthermore, QPLEX~\cite{wang2020qplex} employed a duplex dueling architecture to encode the IGM condition via a neural network, resulting in a more sophisticated value factorization framework. Recent years have also seen increased interest in extending ideas from single-agent reinforcement learning, including distributional RL and risk-sensitive RL, to multi-agent environments~\cite{qiu2021rmix, son2022disentangling, sun2021dfac, shen2023riskq}. Additionally, role-based learning methods have been proposed to effectively decompose tasks in multi-agent systems~\cite{wang2020roma, wang2021rode, zeng2023effective, liu2022rogc}, and the concept of bidirectional dependencies has been explored to enhance cooperative behaviors among agents ~\cite{li2023ace}.

\paragraph{Episode length adjustment}

The problem of efficiently managing episode lengths and their impact on reinforcement learning (RL) has gained increasing attention in recent years. Various approaches have been proposed to address issues related to time limits and trajectory truncation. 
\citet{pardo2018time} explored time limits in RL and introduced methods for handling them effectively, showing that including the remaining time as part of the agent's input improves policy learning by avoiding state aliasing and instability. 
The advantages of truncating trajectories for enhancing Policy Optimization via Importance Sampling (POIS) accuracy were illustrated by~\cite{poiani2023truncating}.
\citet{poiani2024truncating} proposed an adaptive trajectory truncation strategy called RIDO, which allocates interaction budgets dynamically to minimize variance in Monte Carlo policy evaluation. Their results indicate that adaptive schedules outperform fixed-length approaches across multiple RL domains.
In addition to directly adjusting the episode length, there have also been studies addressing uncertain episode lengths.
\citet{mandal2023online} tackled online RL with uncertain episode lengths, proposing an algorithm that adapts to variable horizons by estimating the underlying distribution. They provided theoretical guarantees that highlight the effectiveness of adaptive horizon management in minimizing regret.
Horizon Regularized Advantage (HRA) was proposed by~\cite{nafi2024policy} to improve the generalization of RL. 
By averaging advantage estimates from multiple discount factors, HRA improves policy generalization, particularly in environments with stochastic episode durations.
As described so far, there have been studies related to episode length adjustment in single-agent environments; however, in multi-agent environments, there has been little related research.

\section{Background}

\subsection{Dec-POMDP}

A decentralized partially observable Markov decision process (Dec-POMDP) is a framework used to model multi-agent decision-making in environments where agents must cooperate under conditions of uncertainty and limited observability. 
In a Dec-POMDP, multiple agents interact with a shared environment, but each agent has only partial information about the global state and must make decisions based on its local observations. Formally, a Dec-POMDP can be represented as a tuple 
\(\langle I, S, A, P, O, Z, R, \gamma \rangle\),
where \( I \) is the set of agents, \( S \) represents the set of global states, \( A \) is the set of joint actions across all agents, \( P \) is the state transition function, \( O \) is the set of joint observations, \( Z \) is the observation function, \( R \) is the reward function, and \( \gamma \) is the discount factor.

In this framework, each agent $i \in I:= {1, \ldots, I}$ selects actions $a_{i} \in A$ at each time step $t$.
The state transition function provides the probability of transitioning from one state to another given the joint action of all agents, as \( P(s'|s,\mathbf{a}): S \times A \times S \longmapsto [0,1] \).
Each agent has individual and partial observations $z \in Z$, according to the observation function $O(s, i): S \times I  \longmapsto Z$.
The reward function \( R(s, \mathbf{a}): S \times A \longmapsto \mathbb{R}\) assigns a numerical value to each state-action pair, and $\gamma \in [0,1)$ is the discount factor.
The action-observation history for agent $i$ is $\tau_{i} \in T := (Z \times \mathbf{A})^{*}$, on which its policy $\pi_{i}(a_{i}|\tau_{i}): T \times \mathbf{A} \longmapsto [0, 1]$ is based.
The formal goal is to find a joint policy $\pi$ that optimizes a joint action-value function, which is represented as 
$Q_{jt}^{\boldsymbol\pi}(s, \mathbf{a}) = \mathbb{E}_{s_{0:\infty}, \mathbf{a}_{0:\infty}}[\sum_{t=0}^{\infty}\gamma^{t}r_{t} | s_{0} = s, \mathbf{a}_{0} = \mathbf{a}, \boldsymbol\pi]$.

\subsection{Episode length adjustment}

In reinforcement learning (RL), the agent interacts with the environment in a sequence of steps called an episode, which often has a fixed length \( E_{L} \). However, adjusting episode length during training can significantly influence learning efficiency and policy performance, particularly in environments where the complexity or difficulty varies over time.

The goal of RL is to learn a policy \( \pi \) that maximizes the expected cumulative reward, also known as the return \( G_l \), which is given by:
\begin{align}
G_l = \sum_{k=0}^{E_{L}-1} \gamma^k r(s_{l+k}, a_{l+k}),
\end{align}
where \( E_{L} \) is the length of the episode and $l$ represents the time step within an episode. 
Note that the term 'time step' may cause some confusion with training steps or other terms, so the time step within an episode is referred to as an interaction step $l$.
$l$ is a natural number ranging from $1$ to the episode length, $E_{L}$.
Traditionally, the episode length $E_{L}$ is a fixed value, but adjusting the episode length means setting it to some value smaller than $E_{L}$.

\section{Method}

\begin{algorithm*}
\caption{Adaptive Episode Length Adjustment (AELA)}
\label{algorithm:aela}
Initialize each agent's policy or Q-value network parameters $\theta_1, \theta_2, \dots, \theta_I$, the centralized critic network parameters $\phi$, an empty replay buffer $\mathcal{D}$, the initial episode length $E_{L_0}$, an empty list $\mathcal{L}_H$ for storing entropy values, the window size $w$ for collecting entropy values, and the batch size $B$

Set the current episode length as the initial episode length $E_{L} \gets E_{L_0}$

Set target network parameters $\bar{\theta}_i \gets \theta_i$ for $i = 1, 2, \dots, I$, and $\bar{\phi} \gets \phi$\

Set target network update interval $T_U$

Set $E_{max}$ to the episode length of the original task

Initialize both the time step $t$ and the previous time step $t_{p}$ to $0$

\While{$t < t_{max}$}{

    Initialize the interaction step $l$ to $0$
    
    Initialize the state as the initial state $s_{0}$

    \While{$l < E_{L}$}{    
        Each agent $i$ takes action $a_{i,l} \sim \pi_{\theta_i}(\cdot | o_{i,l})$\
        
        Step into state $s_{l+1}$\
        
        Receive reward $r_{l}$ and observe $o_{i,l+1}$\
        
        Add transition data to the replay buffer: $\mathcal{D} \gets \mathcal{D} \cup \{(s_{l}, o_{l}, a_{l}, r_{l}, s_{l+1}, o_{l+1})\}$\
        
        $l = l + 1$ , $t = t + 1$
        }
        \If{$|\mathcal{D}| > B$}{
            Sample a random batch of episodes from $\mathcal{D}$\
        
            Calculate the total entropy of q-values $H_{\text{total}}$
    
            Add the entropy $H_{\text{total}}$ to $\mathcal{L}_H$\

            Update the parameters of the centralized critic $\phi$ with the sampled batch
        
            Update each agent's policy or Q-value network parameters $\theta_1, \theta_2, \dots, \theta_I$ with the sampled batch

            \If{$|\mathcal{L}_H| \bmod w = 0$}{
                Fit a linear model to the last $w$ values in $\mathcal{L}_H$\
                
                \If{(slope of linear model $< 0$) \textbf{and} ($E_{L} < E_{max}$)}{
                    $E_{L} \gets E_{L}+1$
                }
            }

            \If{$(t - t_{p}) > T_{U}$}{
                Update the target network of centralized critic $\bar{\phi} \gets \phi$, and target networks of agents $\bar{\theta}_1 \gets \theta_1, \dots, \bar{\theta}_I \gets \theta_I$\
                
                $t_{p} = t$
            }
        }
}
\end{algorithm*}

In reinforcement learning, dead-end states refer to states from which the agent cannot reach the goal state or recover once reached. To efficiently collect experience in reinforcement learning, it is necessary to minimize the visits to dead-end states. In single-agent reinforcement learning, several techniques have been proposed to estimate and avoid such dead-ends~\cite{fatemi2019dead, zhang2023safe, killian2023risk}. 
However, in multi-agent reinforcement learning, the number of states and actions grows exponentially, and the states that lead to successful rewards are relatively few compared to the total number of states and actions. This makes it extremely challenging to directly estimate dead-ends.

Therefore, instead of precisely estimating and completely avoiding dead-ends, this paper focuses on efficiently collecting experience even in the presence of dead-ends. In multi-agent environments where dead-ends exist, effective experience collection often involves visiting states that occur before reaching a dead-end. This strategy is referred to as secure exploration, which emphasizes maximizing the exploration of states that precede dead-ends to ensure meaningful learning opportunities. 

\subsection{Theoretical background}
\label{subsection:thm}
In this section, we theoretically describe the relationship between reducing episode length and visiting secure states. 
First, we define a dead-end state and a secure state as follows:


\begin{definition}[Dead-end State]
Once the agent reaches dead-end states $s_d$ at time step $t$ in a trajectory, no policy $\pi$ can lead it to the goal state $s_g$ for any $t'>t$.

\end{definition}

\begin{definition}[Secure state]
A state that is not a dead-end state is referred to as a secure state $s_{s}$, as there exists at least one policy that guarantees reaching the goal state.

\end{definition}
According to these definitions, the properties of $P_{s}$ with respect to the interaction step can be proven as follows.
\begin{lemma}
Let $E_L$ denote the episode length, and let $l$ (where $l = 1, 2, \dots, E_L$) denote the interaction step (i.e., the time step within each episode). Let $P_s(l)$ be the probability that the state is secure at interaction step $l$. Then, it holds that:
\begin{equation}
P_{s}(l+1) \leq P_{s}(l).
\end{equation}
\label{lemma:mono}
\end{lemma}

\begin{proof}[\textbf{Proof}]
Let the probability of falling into a dead-end state at time step \( l \) be denoted by \( P_{d}(l) \). 
To remain in a secure state up to time step \( l \), the agent must avoid falling into a dead-end state at every previous time step. 
Thus, we have:
\begin{align}
P_{s}(l) = \prod_{k=1}^{l} (1 - P_{d}(k)).
\end{align}
To prove that \( P_{s}(l) \) is not increasing with respect to \( l \), we compare \( P_{s}(l) \) and \( P_{s}(l+1) \):
\begin{align}
P_{s}(l+1) = P_{s}(l) \times (1 - P_{d}(l+1)),
\end{align}
where \( 1 - P_{d}(l+1) \) represents the probability of avoiding a dead-end state at time step \( l+1 \). Since \( 0 \leq 1 - P_{d}(l+1) \leq 1 \), it follows that:
\begin{align}
P_{s}(l+1) = P_{s}(l) \times (1 - P_{d}(l+1)) \leq P_{s}(l).
\end{align}
Hence, by comparing \( P_{s}(l) \) and \( P_{s}(l+1) \), it can be concluded that \( P_{s}(l) \) is a non-increasing function of \( l \). 
Therefore, the probability of being in a secure state decreases monotonically or remains constant as the interaction step $l$ increases.
\end{proof}

Based on Lemma \ref{lemma:mono}, we prove that reducing the episode length increases or maintains the probability of being a secure state, thereby enabling efficient experience collection.
The following two theorems build upon Lemma \ref{lemma:mono} to establish the theoretical foundation for our approach. Specifically, Theorem \ref{theorem:increase} provides a theoretical hint that effectively adjusting the episode length can yield benefits (proof in Appendix A). Therefore, in AELA, we implement adaptive control of the episode length based on this insight. Subsequently, Theorem \ref{theorem:regret} elucidates the impact of decreasing the probability of visiting dead-end states and the benefits in terms of regret minimization. These two theorems provide the theoretical foundation for the design of the AELA algorithm.
\begin{theorem}
Let $E_L$ be the fixed episode length, and $N_{s}$ be the expected number of secure state visits over all episodes.
Then, $N_{s}$ remains constant or increases if $E_L$ is reduced.
\label{theorem:increase}
\end{theorem}
\begin{definition}[The probability of visiting dead-end states]
\label{definition:p_d}
Let $N_{\text{total}}$ denote the total number of collected data samples.
When \( N_d \) denotes the expected number of dead-end state visits across all episodes, then \( N_d \) can be expressed as \( N_{\text{total}} - N_{s} \). The probability \( P_d \) of visiting dead-end states can be defined as \( \frac{N_d}{N_{\text{total}}} \).
\end{definition}

\begin{definition}[Regret]
\label{definition:regret}
In the context of reinforcement learning, the regret after $T$ time steps, denoted by $\text{Regret}(T)$, measures the difference in cumulative rewards between the optimal policy $\pi^*$ and the agent's policy $\pi$:
\begin{equation}
\text{Regret}(T) = R(\pi^*) - R(\pi),
\end{equation}
where \( R(\pi) = \sum_{t=1}^{T} r_t \) and \( R(\pi^*) = \sum_{t=1}^{T} r_t^* \).
Here, $r_t$ represents the reward obtained by the agent's policy $\pi$ at time step $t$, and $r_t^*$ represents the reward obtained by the optimal policy $\pi^*$ at the same time step.
\end{definition}
\begin{corollary}
\label{corollary:pd_decrease}
According to Theorem \ref{theorem:increase}, reducing the episode length \( E_L \) either increases or preserves the number of visits to secure states \( N_{s} \).
Consequently, reducing \( E_L \) will cause \( P_d \) to either decrease or at least remain unchanged.
\end{corollary}

\begin{proof}[\textbf{Proof}]

From Theorem \ref{theorem:increase}, we know that reducing \(E_L\) causes \(N_s\) to either increase or remain unchanged.
Given that \( P_d = \frac{N_d}{N_{\text{total}}} = 1 - \frac{N_{s}}{N_{\text{total}}} \) (defined in Definition \ref{definition:p_d}), a non-decreasing change in \( N_s \) (i.e., an increase or no change) results in a non-increasing change in \( P_d \) (i.e., a decrease or no change).


\end{proof}

Based on Corollary \ref{corollary:pd_decrease}, we establish the following Theorem \ref{theorem:regret} (proof is provided in Appendix B). To do so, we introduce an assumption that ensures the agent prioritizes reaching the goal state over accumulating intermediate rewards. 

\begin{assumption}
The reward \( r_g \) obtained at the goal state is sufficiently large compared to the rewards received over any interval of time steps within the episode. Formally, for any \( 1 \leq k \leq l \leq T \),
\begin{align}
\sum_{t=k}^{l} r_t < r_g.
\end{align}
This assumption ensures that the agent prioritizes reaching the goal state over accumulating intermediate rewards.
\label{assumption:goal_reward}
\end{assumption}

\begin{theorem}
Under Assumption \ref{assumption:goal_reward}, as \( P_d \) decreases, the difference in cumulative rewards between the optimal policy and the agent's policy also decreases, thereby reducing Regret(T) (defined in Definition \ref{definition:regret}).
\label{theorem:regret}
\end{theorem}

Combining Theorem \ref{theorem:increase} and Theorem \ref{theorem:regret}, we conclude that reducing the episode length decreases the probability of entering a dead-end state $P_d$, which in turn reduces the regret term, except when the episode length reduction preserves $P_d$. 
Building upon the theoretical insights established in this Section \ref{subsection:thm}, we introduce the Adaptive Episode Length Adjustment (AELA) algorithm in the following section. 
\color{black}

\subsection{Algorithm : AELA}

AELA leverages the relationship between episode length and secure exploration, as demonstrated by Theorems \ref{theorem:increase} and \ref{theorem:regret}, to optimize learning efficiency. By adaptively controlling the episode length based on the agent's learning progress, AELA effectively minimizes the probability of encountering dead-end states, thereby reducing regret and enhancing overall performance in multi-agent environments. To summarize, we realized the advantages of training with reduced episode lengths.
However, starting with shorter episodes requires adjusting their length as the agent needs to gather experiences beyond the reduced episode lengths. To address this, we propose an adaptive method within the AELA framework that incrementally increases the episode length based on the agent's learning progress. Specifically, the algorithm monitors the entropy of the learned policy or Q-values; a decreasing trend in entropy indicates that the agent's policy is converging toward stability. When such convergence is detected, the episode length is extended to allow the agent to gather more experiences beyond the initial secure exploration phase. 

Given a set of Q-values \( Q(s, a) \) for each action \( a \) in a state \( s \), the action probability distribution \( P(a) \) can be defined using the softmax function:
\begin{equation}
P(a \mid s) = \frac{\exp(Q(s, a) / \tau)}{\sum_{a} \exp(Q(s, a) / \tau)},
\end{equation}
where \( \tau \) is the temperature parameter that controls the smoothness of the action probability distribution.
If a policy-based method is used, the action probability distribution represents the policy $\pi$.
The entropy \( H \) of the action distribution can then be calculated as:
\begin{equation}
H = - \sum_{a} P(a \mid s) \log P(a \mid s).
\end{equation}
Then, we sum all the entropy values over a mini-batch:
\begin{equation}
H_{\text{total}} = \sum_{b=1}^{B} \sum_{t=1}^{T} \sum_{i=1}^{I} \left( - \sum_{a} P(a_{t}^{(b, i)} \mid s_{t}^{(b, i)}) \log P(a_{t}^{(b, i)} \mid s_{t}^{(b, i)}) \right),
\end{equation}
where, \( H_{\text{total}} \) represents the total entropy accumulated across all batches, time steps, and agents. The variable \( B \) refers to the batch size, which indicates the number of independent episodes being processed simultaneously. The variable \( T \) denotes the number of time steps within each episode, and \( I \) is the number of agents in the multi-agent system.
The action taken by agent \( i \) at time step \( t \) in batch \( b \) is denoted by \( a_{t}^{(b, i)} \), while the state observed by the same agent at the same time step and batch is represented by \( s_{t}^{(b, i)} \). The probability of taking action \( a \) given the state \( s \) for agent \( i \) at time step \( t \) in batch \( b \) is given by \( P(a_{t}^{(b, i)} \mid s_{t}^{(b, i)}) \).

The entropy \( H_{\text{total}} \) is collected over a window size of \( w \).
If the entropy is decreasing, the episode length is increased by $1$ as:

\begin{equation}
E_L \leftarrow E_L + f(H_{\text{trend}}),
f(H_{\text{trend}}) =
\begin{cases}
    1 & \text{if } \alpha < 0, \\
    0 & \text{otherwise},
\end{cases}
\end{equation}

where \( H_{\text{trend}} = \alpha H_{total} + \beta \) represents the linear fitting function for the entropy values in a window. If the slope \( \alpha \) is negative, \( f(H_{\text{trend}}) \) returns 1; otherwise, it returns 0.
The detailed description of the algorithm is presented in Algorithm \ref{algorithm:aela}.
This adaptive adjustment ensures that the agent benefits from reduced dead-end probabilities during the critical early learning stages while still being able to explore and learn from more complex scenarios in the latter parts of episodes.

\section{Experiment}

In this section, we present the experimental results of AELA under two well-known MARL test-beds: the StarCraft multi-agent challenge (SMAC) and modified predator-prey (MPP).
We evaluated the effectiveness of our proposed algorithm by combining it with two widely known algorithms in MARL: VDN and QMIX. We named these combinations AELA-VDN and AELA-QMIX, respectively.

\begin{figure}[ht]
\centering
\subfigure[$P=-2$]{\includegraphics[height=35mm]{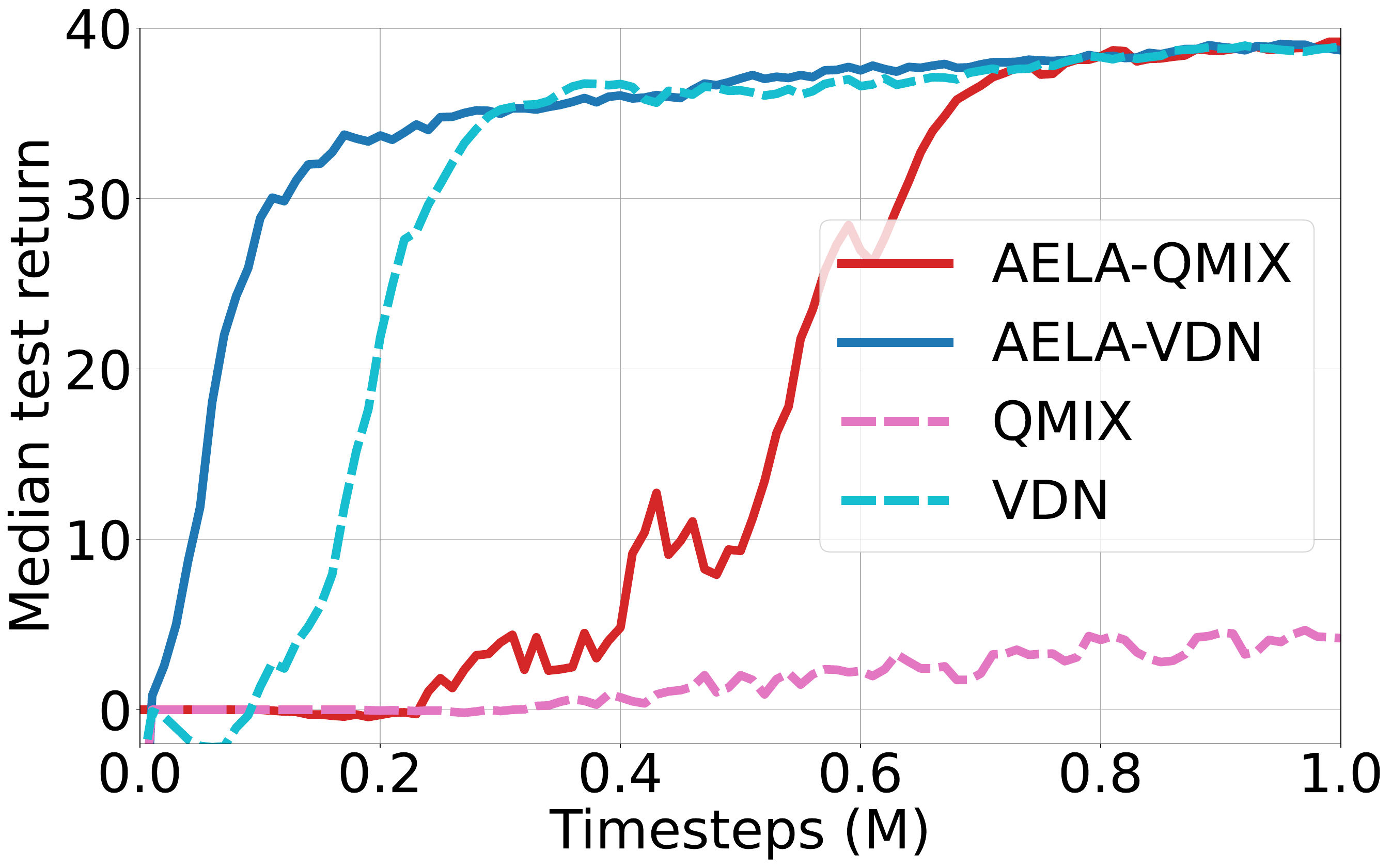}}
\subfigure[$P=-4$]{\includegraphics[height=35mm]{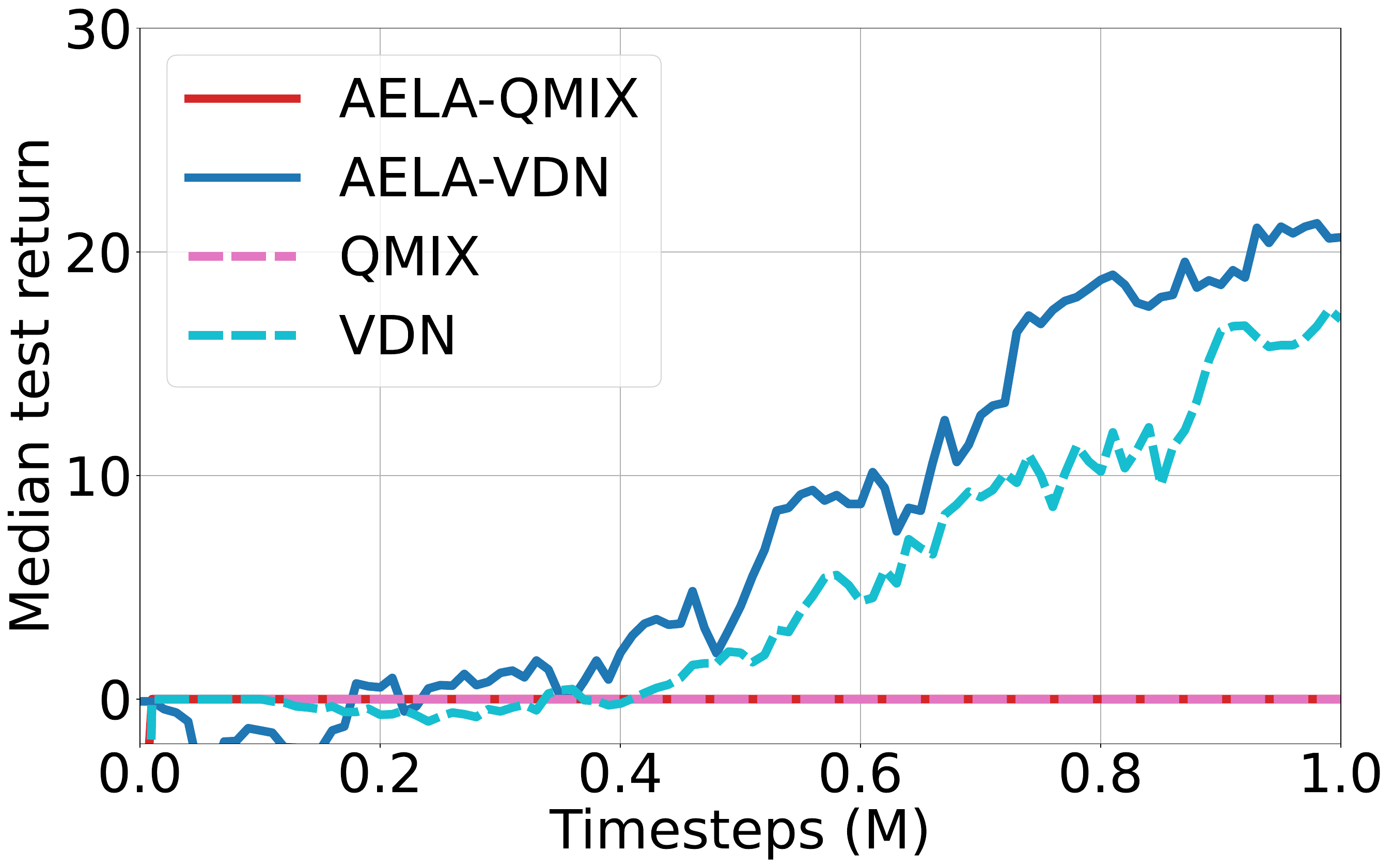}}
\caption{Median test return in the MPP tasks}
\label{fig:mpp}
\end{figure}

\begin{figure*}[ht]
\centering
\includegraphics[width=1.0\columnwidth]{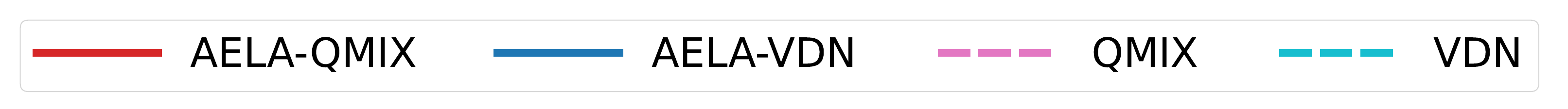} \\
\subfigure[3m]{\includegraphics[height=35mm]{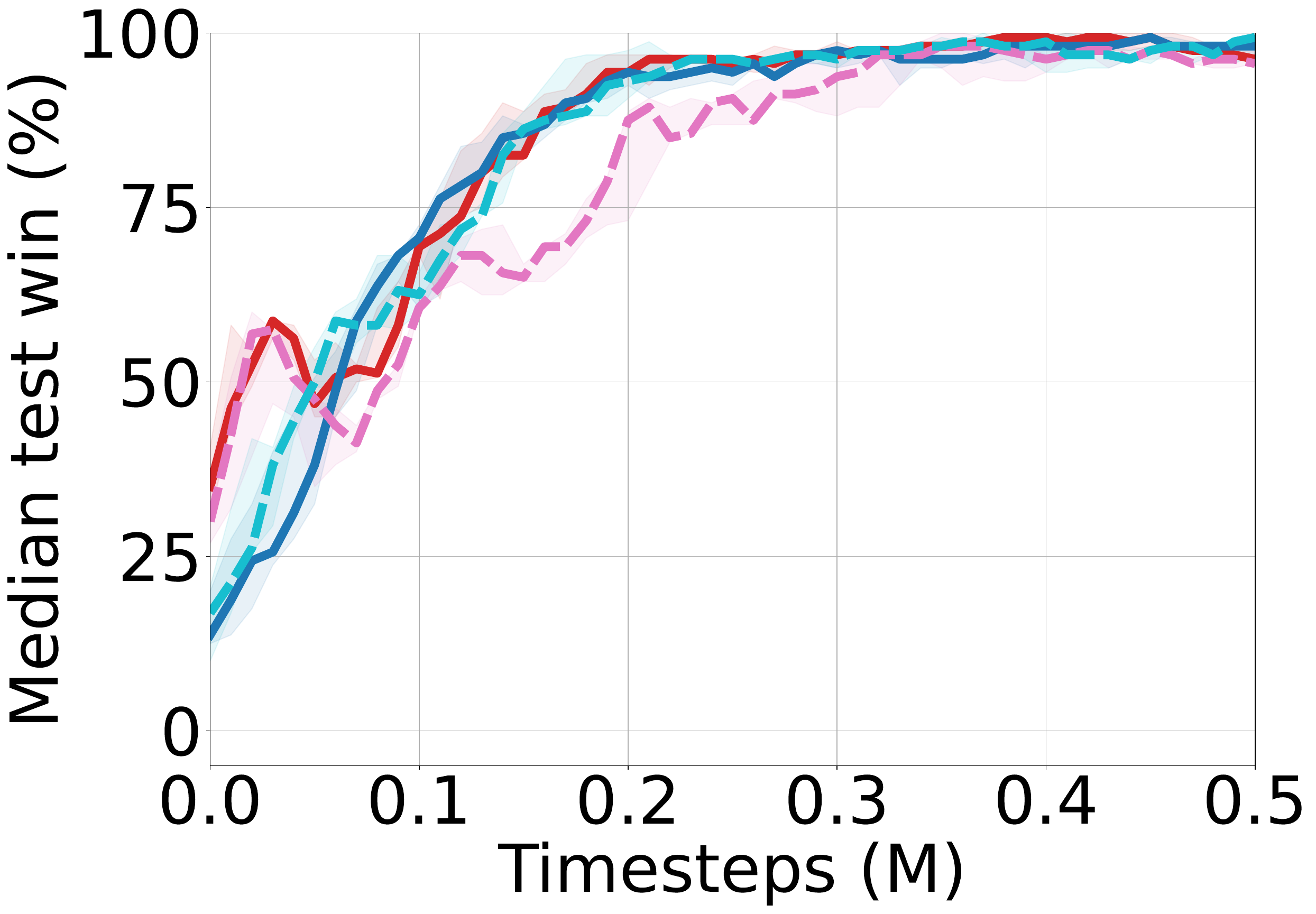}} 
\subfigure[2s\_vs\_1sc]{\includegraphics[height=35mm]{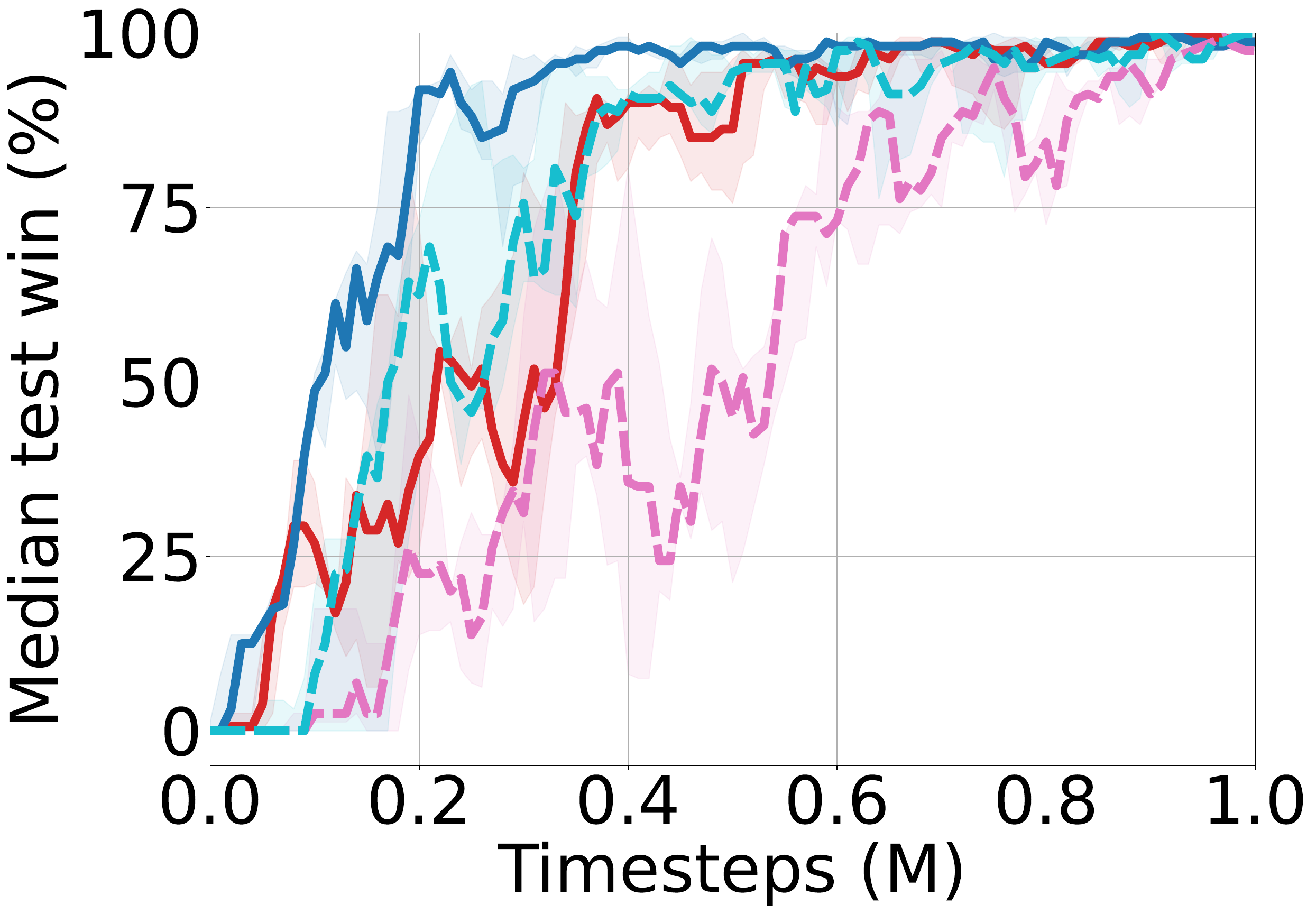}}
\subfigure[MMM2]{\includegraphics[height=35mm]
{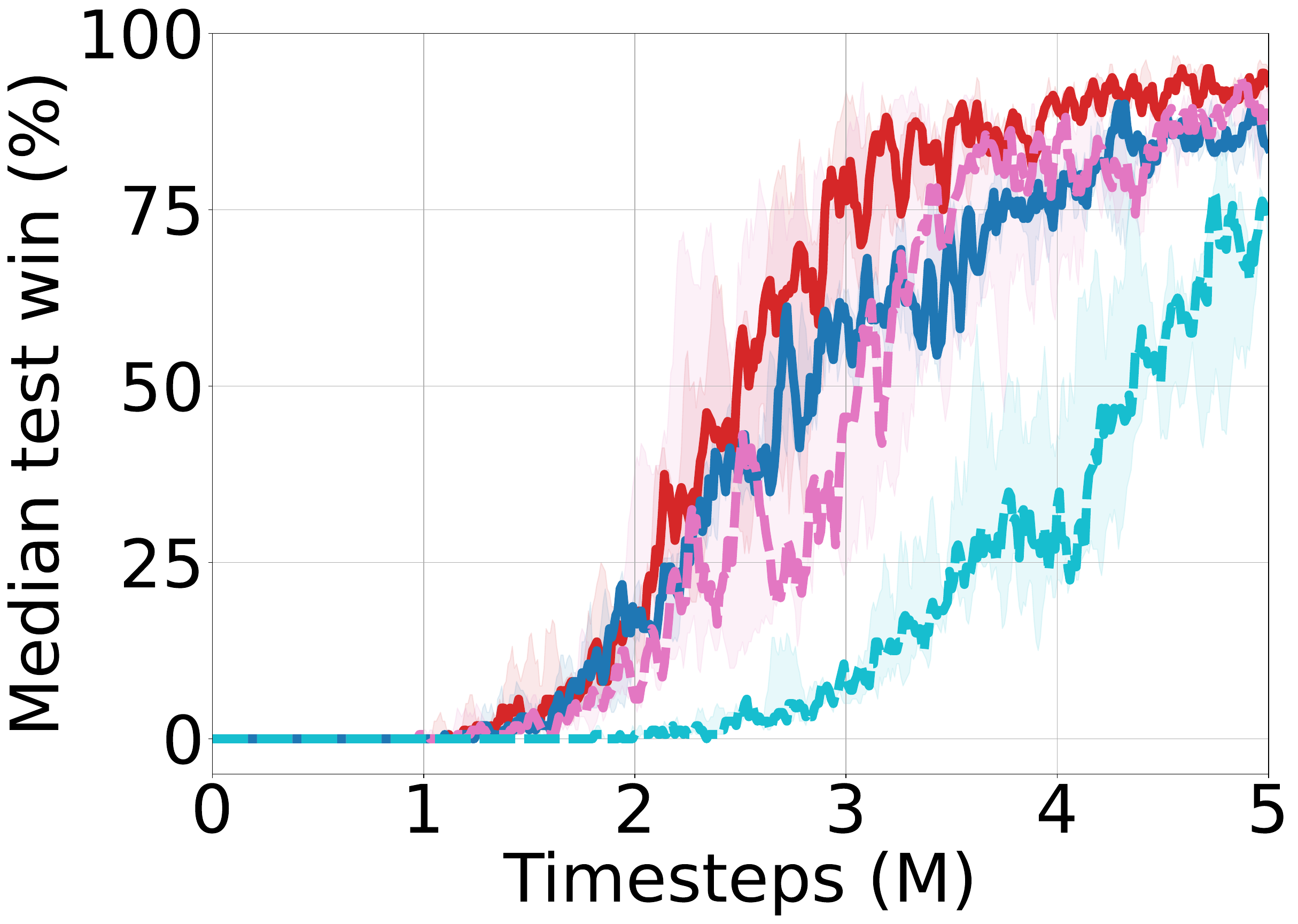}}
\subfigure[3s5z\_vs\_3s6z]{\includegraphics[height=35mm]
{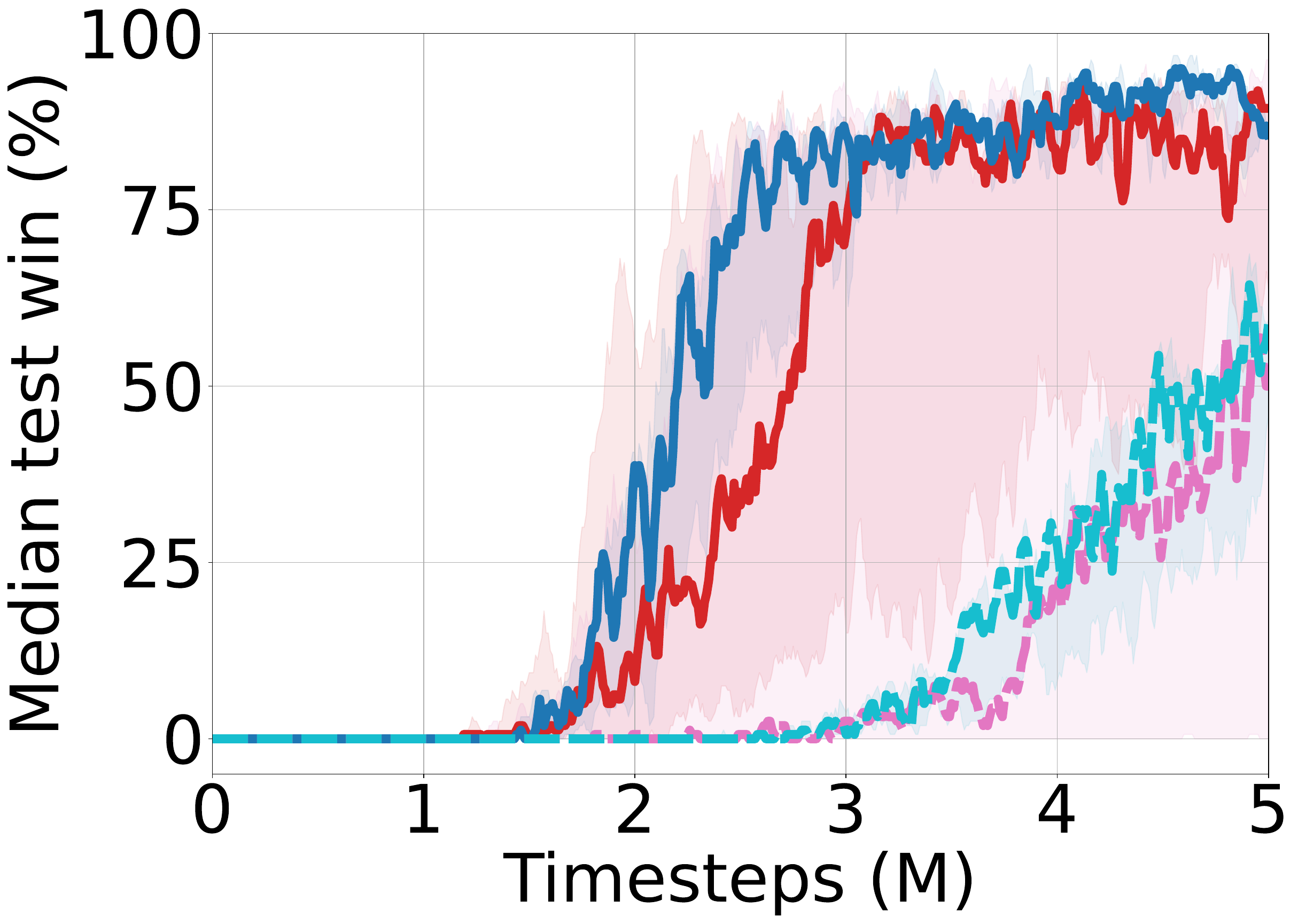}}
\subfigure[6h\_vs\_8z]{\includegraphics[height=35mm]
{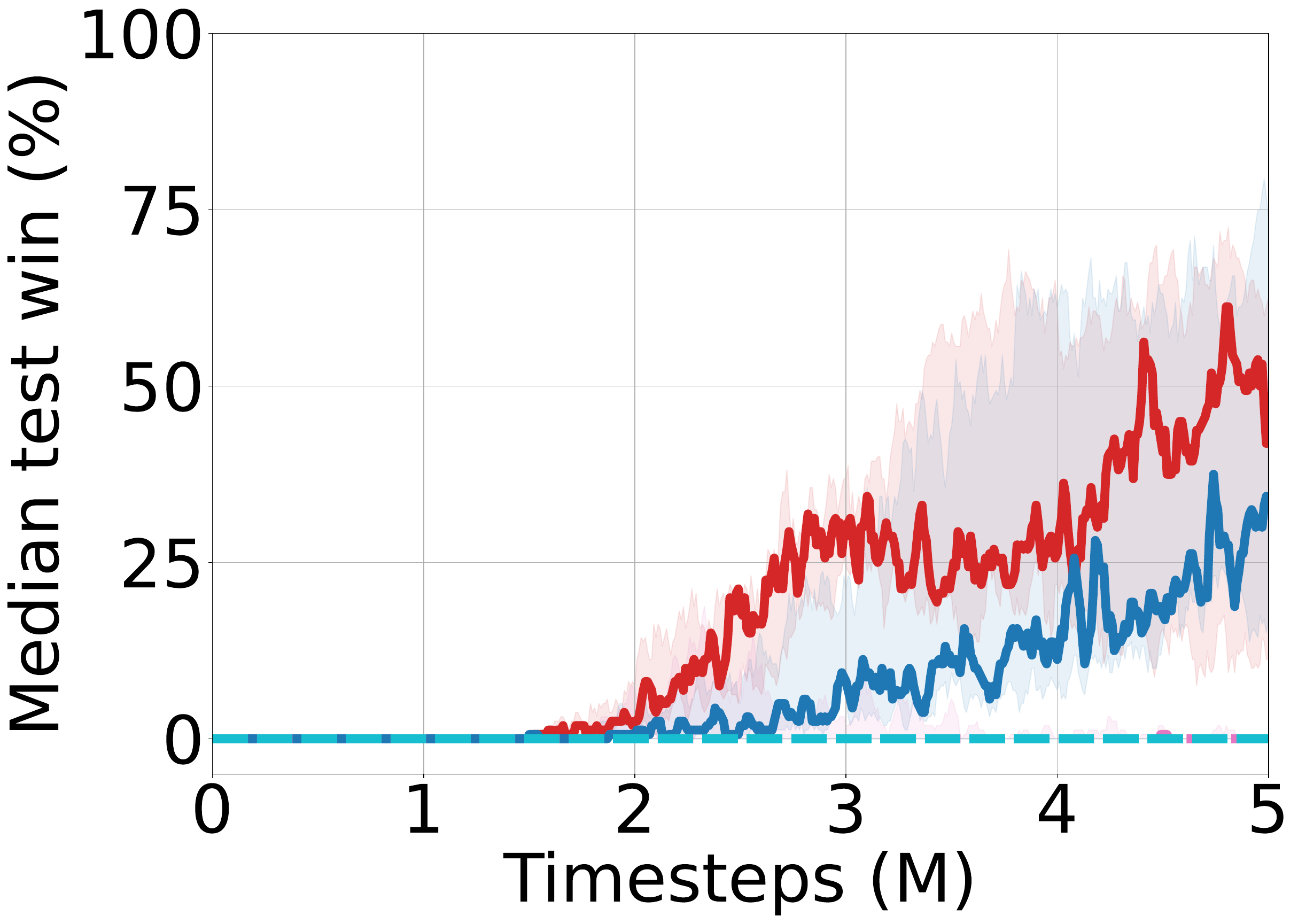}}
\subfigure[Corridor]{\includegraphics[height=35mm]
{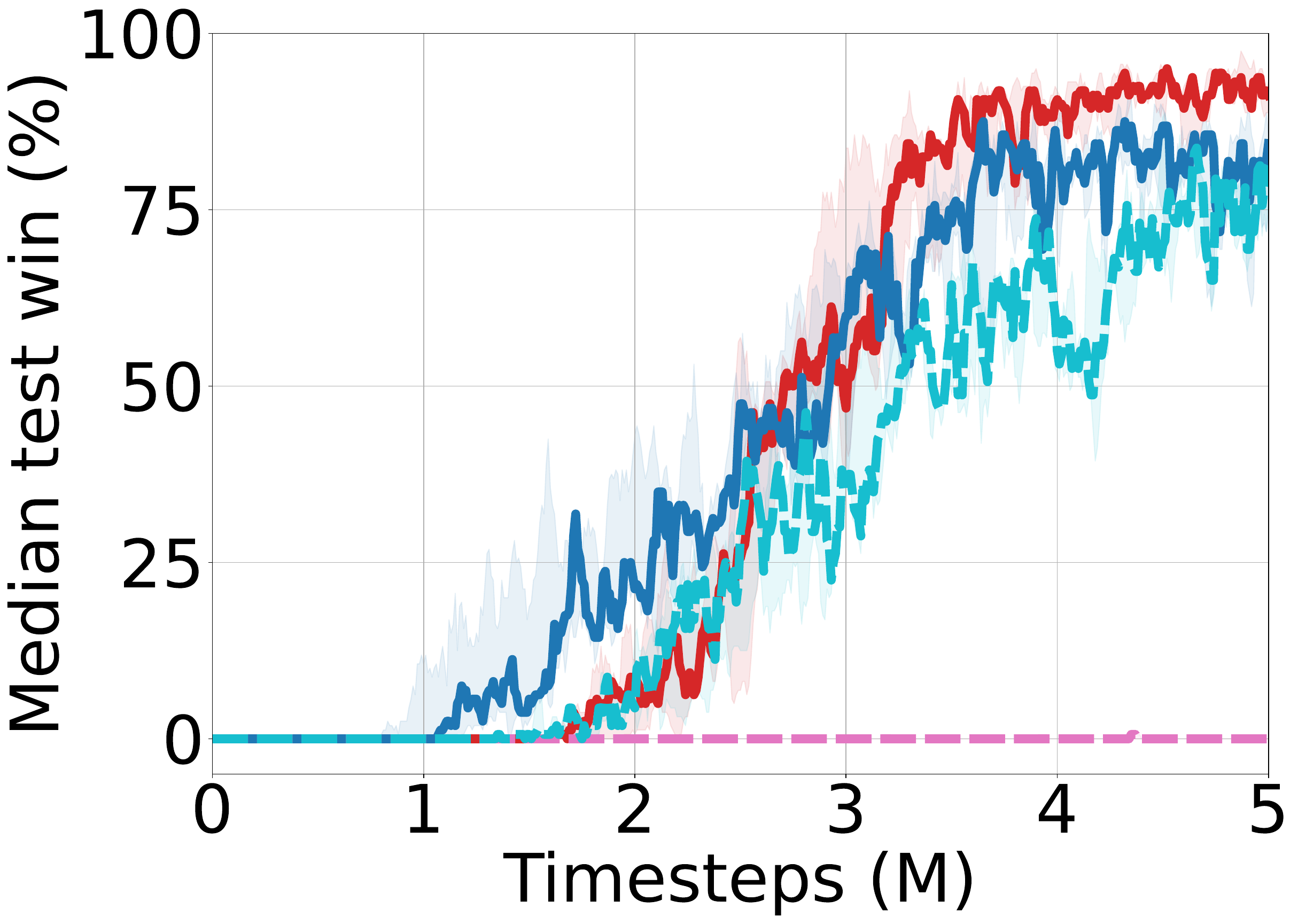}}
\caption{Median test win rates with different SMAC scenarios}
\label{fig:per}
\end{figure*}

\begin{figure}[ht]
\centering
\subfigure[3s5z\_vs\_3s6z]{\includegraphics[height=29mm]{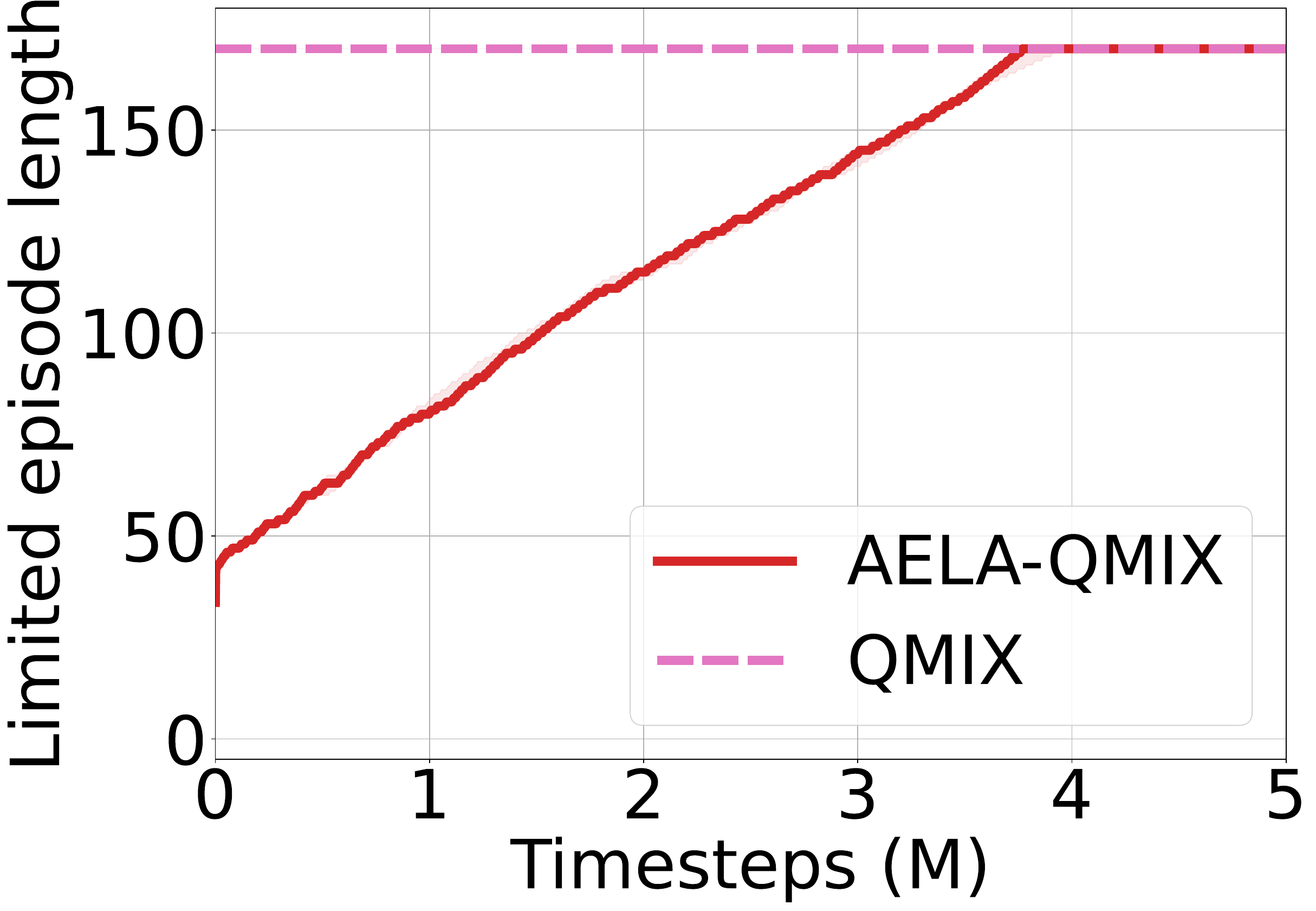}}
\subfigure[6h\_vs\_8z]{\includegraphics[height=29mm]
{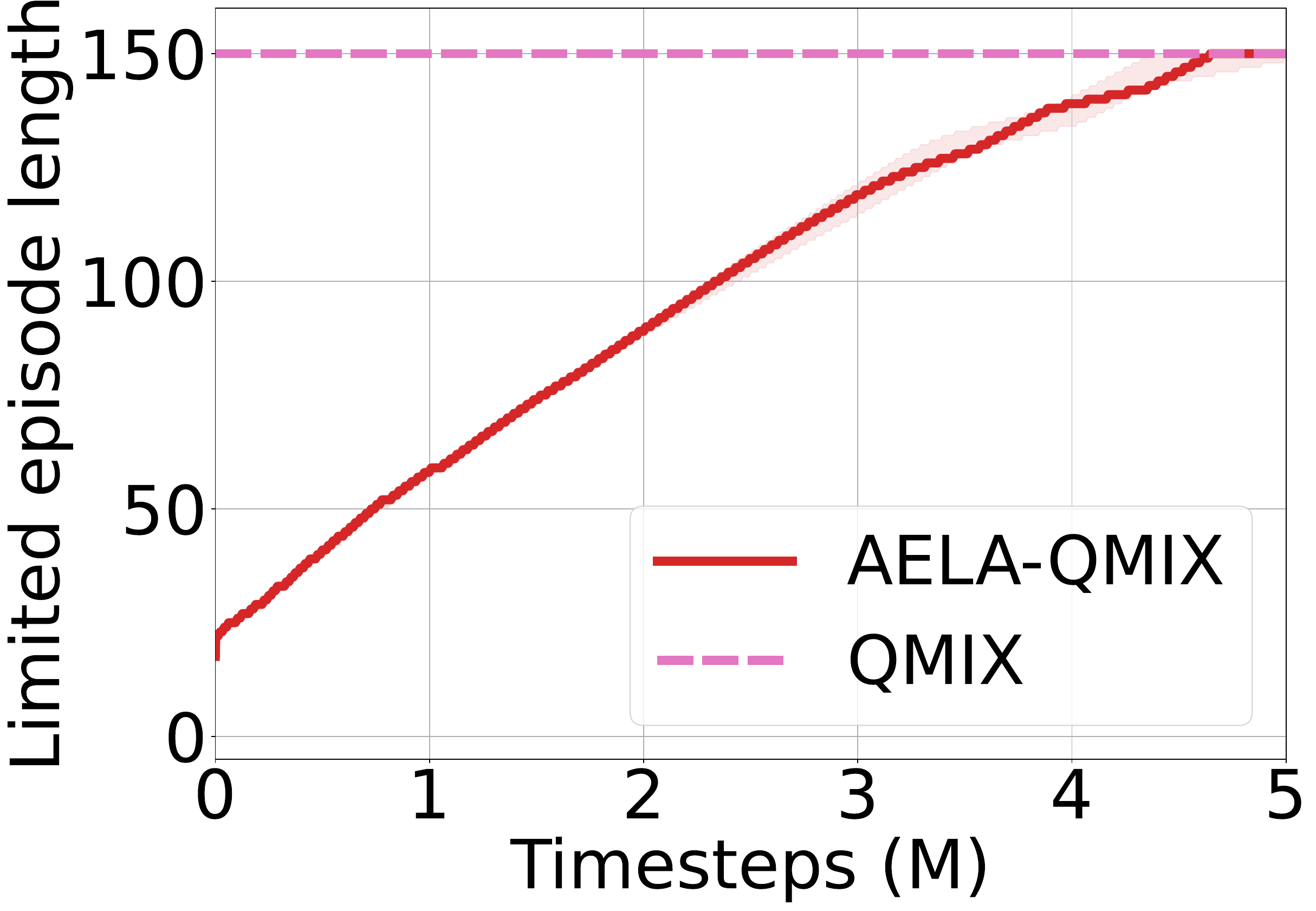}}
\caption{Limited episode length during training}
\label{fig:traineplength}
\end{figure}

\begin{figure}[ht]
\centering
\subfigure[3s5z\_vs\_3s6z]{\includegraphics[height=35mm]{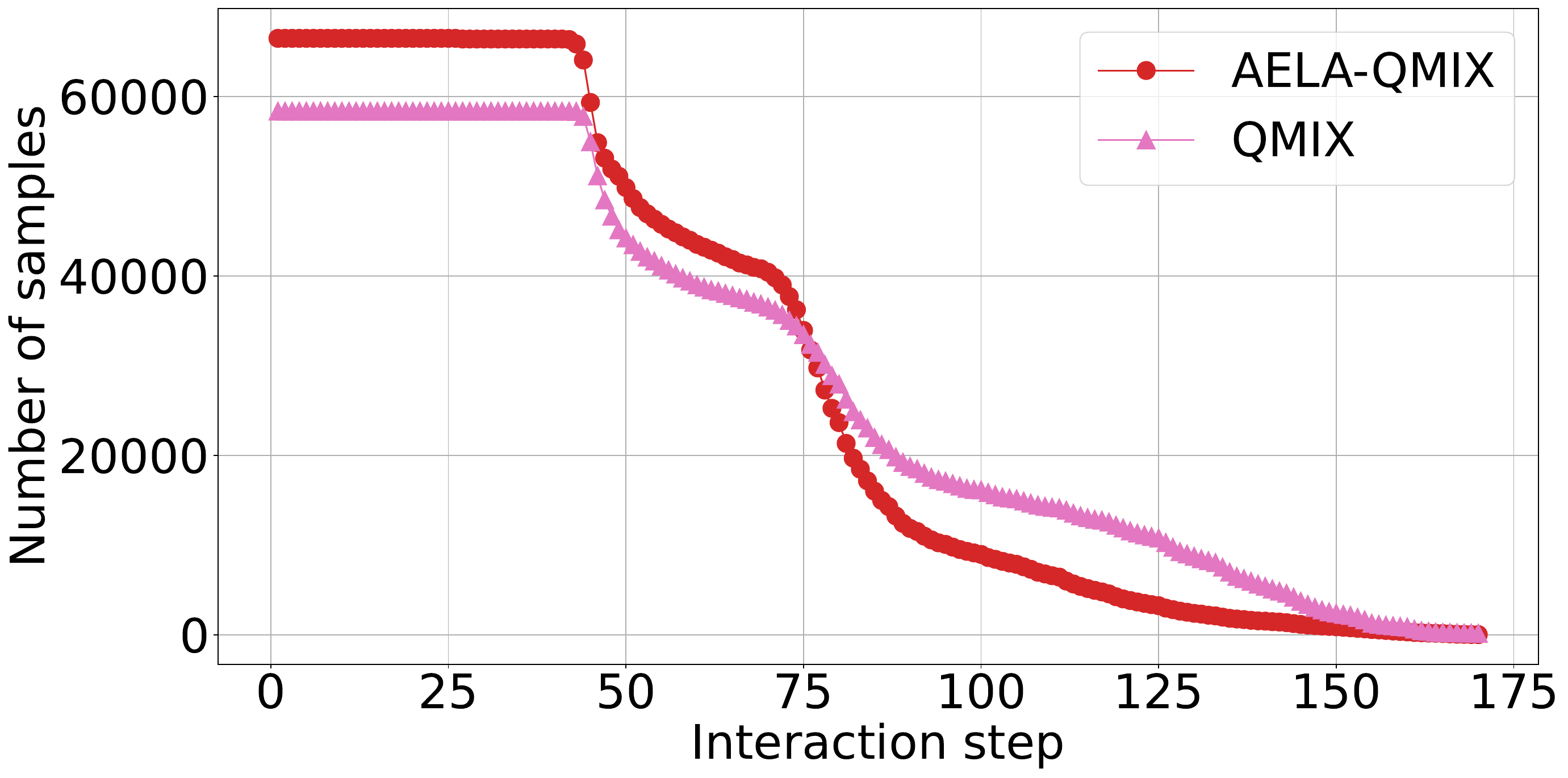}}
\subfigure[6h\_vs\_8z]{\includegraphics[height=35mm]
{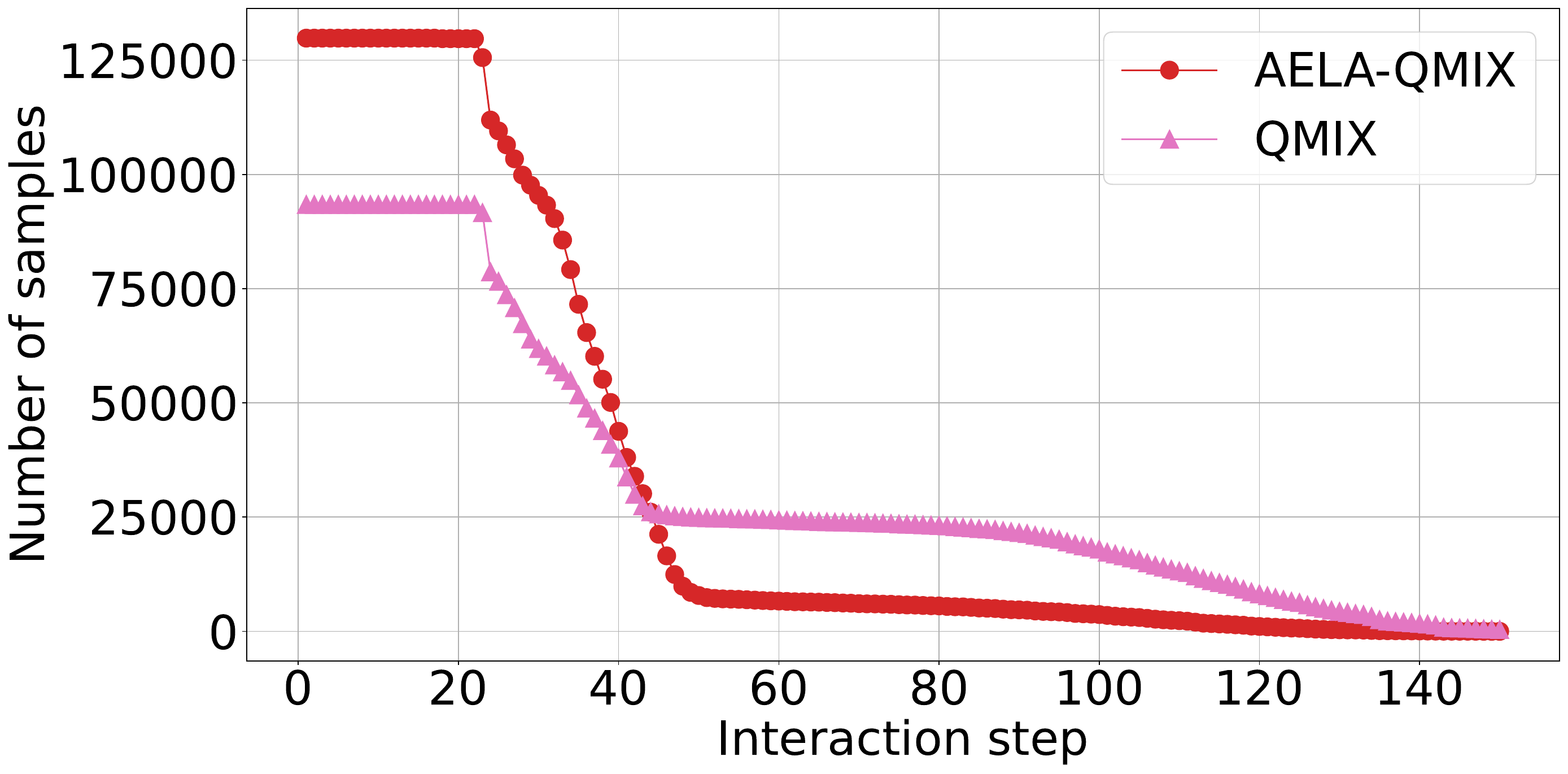}}
\caption{Number of samples with interaction steps}
\label{fig:visit}
\end{figure}

\subsection{Modified predator-prey}
The modified predator-prey scenario is an extended version of the classical predator-prey problem, widely used in MARL research to evaluate agent coordination and learning capabilities. In this scenario, multiple predators collaborate to capture prey within an environment, emphasizing the need for cooperative learning. 
The modified predator-prey scenario employs a more intricate reward structure compared to the original. 
If two or more predators catch the prey, a reward of $+10$ is given, whereas if only one predator catches the prey, a negative reward of $P$ is given.
We tested AELA with eight predators and eight prey, and set $P$ as $-2$ and $-4$.
The experiment is conducted over 16 episodes, and the median return from five independent runs is recorded every 10,000 training steps."

As shown in Figure \ref{fig:mpp}(a), the proposed methods AELA-VDN and AELA-QMIX, which combine AELA with the original algorithms, demonstrate a faster increase in cumulative rewards compared to the original algorithms.
In particular, for QMIX, which is known to be highly sensitive to the order in which states are visited~\cite{mahajan2019maven, gupta2021uneven, zheng2021episodic}, performance significantly improves after combining with AELA due to its ability to prevent the repetition of incorrect actions early in an episode by limiting the episode length.
In the scenario with $P = -4$ depicted in Figure \ref{fig:mpp}(b), only VDN found a policy that achieves high cumulative rewards, and even here, AELA-VDN exhibited a faster rise in cumulative rewards. Although there may be cases where a good policy is not found depending on the value decomposition method, the results in Figure \ref{fig:mpp}(b) are still meaningful as the experimental findings show no degradation in performance, consistent with the theoretical background discussed in Section \ref{subsection:thm}.

\begin{figure*}
\centering
\subfigure[$l=1$]{\includegraphics[width=0.25\textwidth]{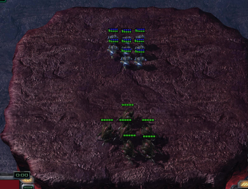}}
\subfigure[$l=7$]{\includegraphics[width=0.25\textwidth]{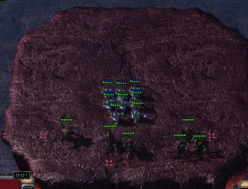}}
\subfigure[$l=15$]{\includegraphics[width=0.25\textwidth]{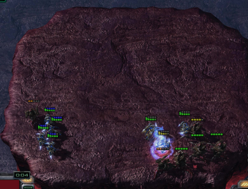}}
\subfigure[$l=22$]{\includegraphics[width=0.25\textwidth]{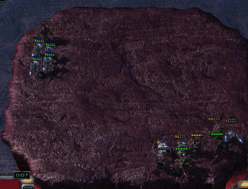}}
\subfigure[$l=33$]{\includegraphics[width=0.25\textwidth]{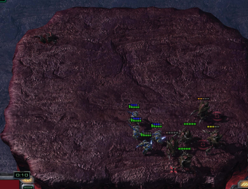}}
\subfigure[$l=58$]{\includegraphics[width=0.25\textwidth]{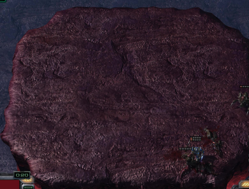}}
\caption{Snapshot of the final policy for AELA-QMIX in 6h\_vs\_8z with interaction steps $l$}
\label{fig:snapaela}
\end{figure*}

\begin{figure}[ht]
\centering
\subfigure[3s5z\_vs\_3s6z]{\includegraphics[height=29mm]{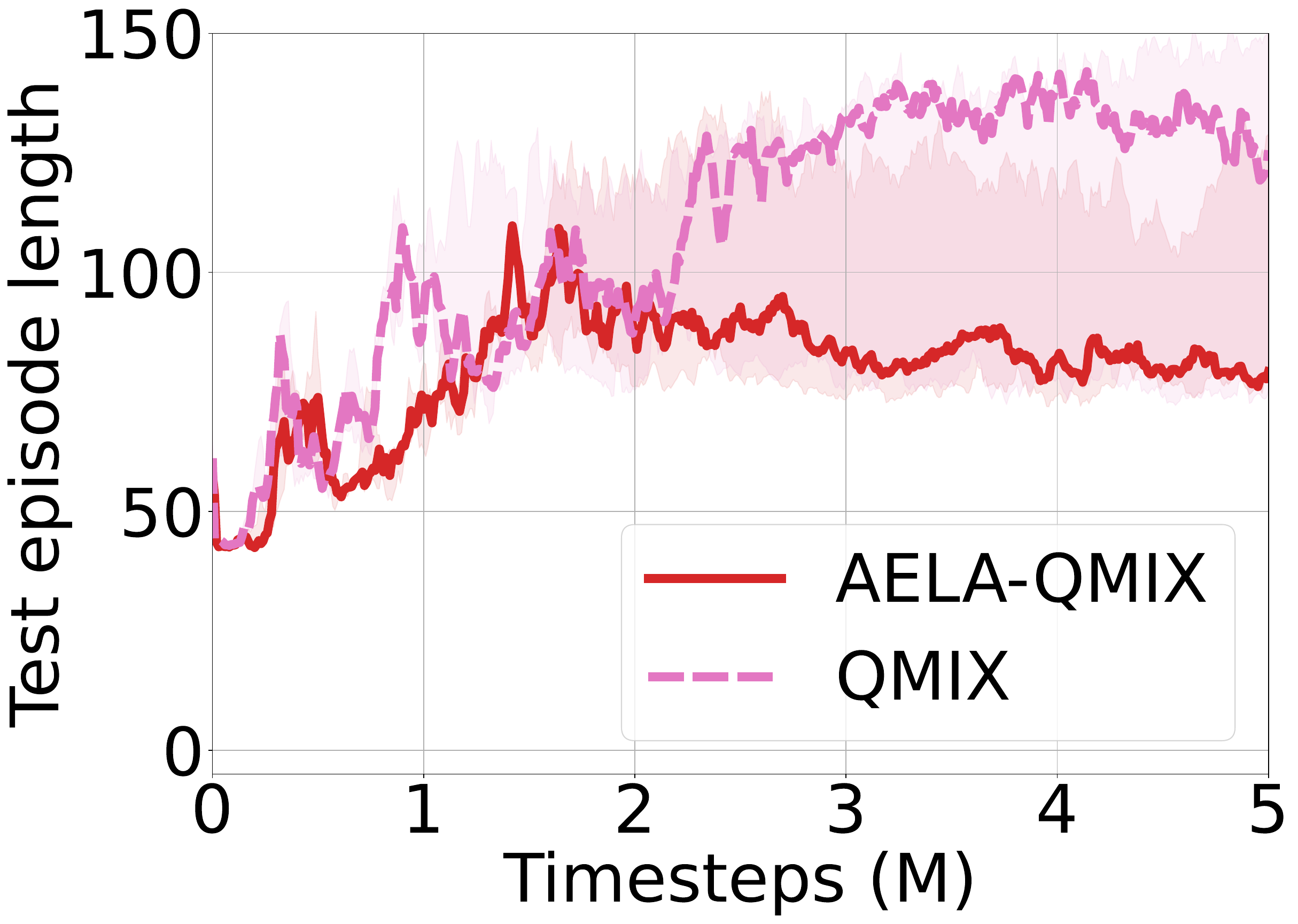}}
\subfigure[6h\_vs\_8z]{\includegraphics[height=29mm]
{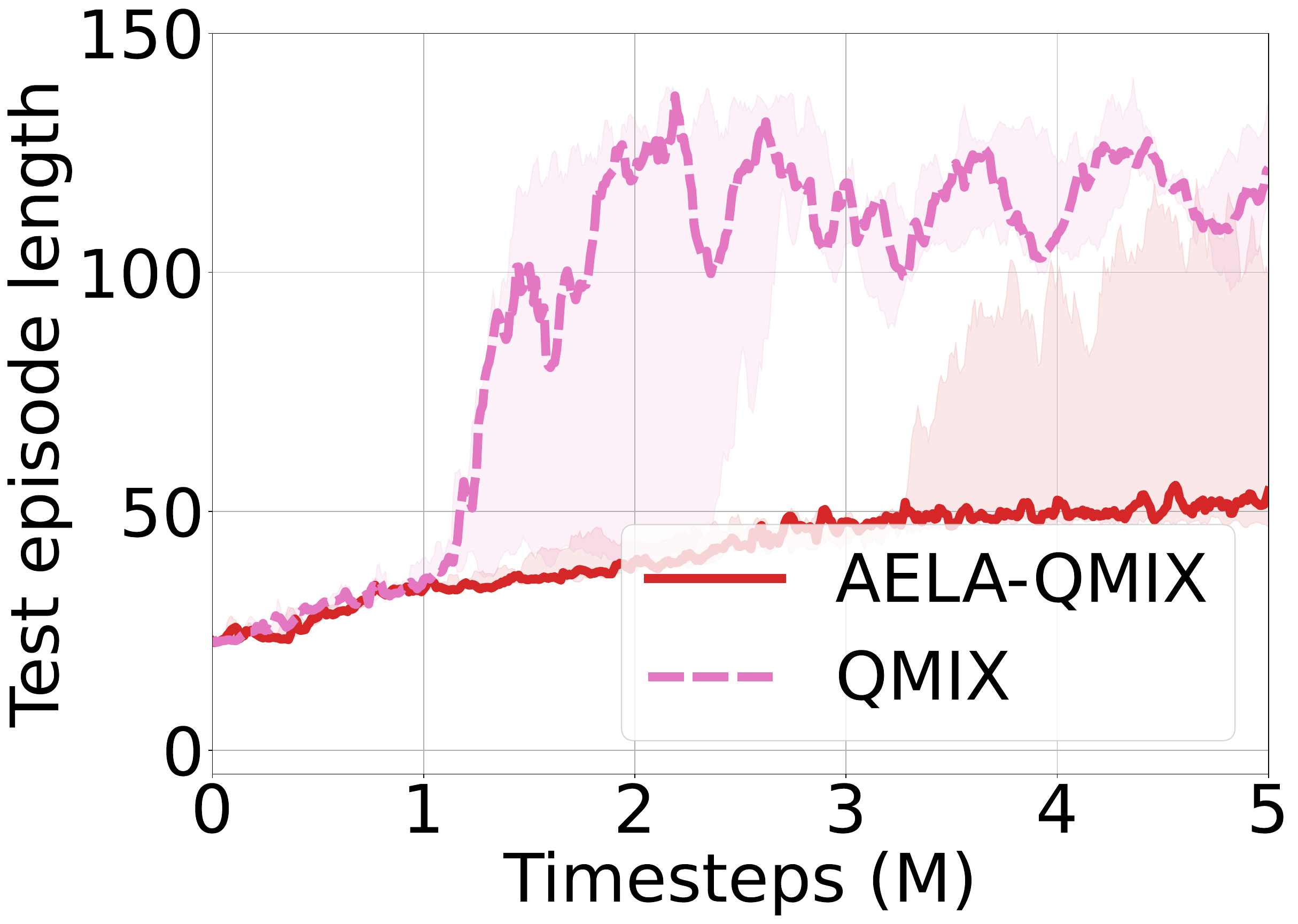}}
\caption{Interaction step at which an episode ends when following the current best policy during testing}
\label{fig:testeplength}
\end{figure}

\begin{figure*}
\centering
\subfigure[$l=1$]{\includegraphics[width=0.25\textwidth]{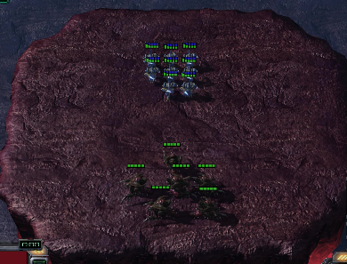}}
\subfigure[$l=7$]{\includegraphics[width=0.25\textwidth]{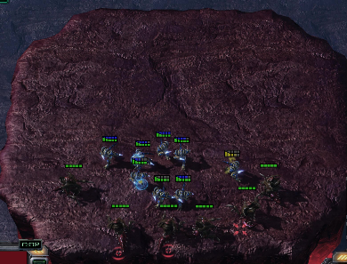}}
\subfigure[$l=21$]{\includegraphics[width=0.25\textwidth]{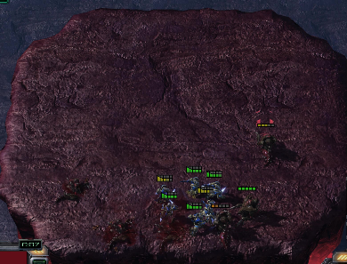}}
\subfigure[$l=38$]{\includegraphics[width=0.25\textwidth]{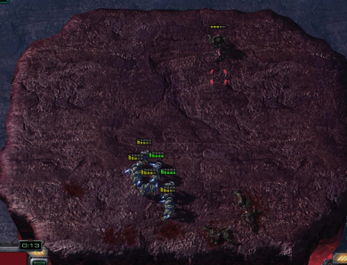}}
\subfigure[$l=75$]{\includegraphics[width=0.25\textwidth]{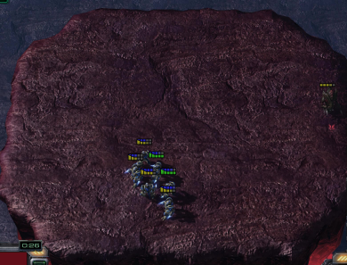}}
\subfigure[$l=149$]{\includegraphics[width=0.25\textwidth]{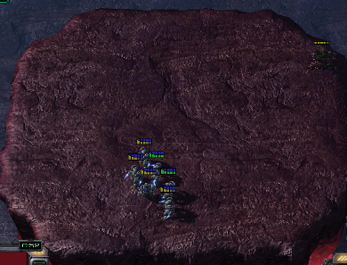}}
\caption{Snapshot of the final policy for QMIX in 6h\_vs\_8z with interaction steps $l$}
\label{fig:snapqmix}
\end{figure*}

\subsection{StarCraft multi-agent challenge}
In SMAC, agents are tasked with solving micromanagement scenarios where they must coordinate their actions to defeat opponent units controlled by the StarCraft II game engine. These scenarios vary in difficulty, unit types, and the number of agents involved, offering various challenges for evaluating MARL methods. The reward function in SMAC is typically defined based on damage dealt, units killed, and victory in the game, thus incentivizing effective coordination and combat tactics. Agents must make decisions with limited observability, requiring efficient exploration and accurate value estimation to achieve optimal performance.
We performed five runs per scenario, and in each run, we conducted 32 tests every 10,000 training steps to calculate the win rate.

Figure \ref{fig:per} shows the median win rate over five runs for VDN, QMIX, AELA-VDN, and AELA-QMIX.
In the relatively easier scenarios, such as \textit{3m} and \textit{2s\_vs\_1sc}, both methods reach a final win rate of nearly 100\%, but the convergence speed is faster when AELA is applied. In the remaining, even more challenging scenarios, AELA leads to improved final performance. Notably, in the \textit{6h\_vs\_8z} and \textit{Corridor} scenarios, the win rate with the original QMIX was close to zero, whereas AELA-QMIX achieved a significantly higher win rate.

Interestingly, performance improvements were achieved solely by adaptively limiting the episode length based on the learning situation without applying additional intrinsic rewards or other complicated value factorization mechanisms.
We analyze that AELA improves performance by focusing on exploring diverse states early in the episode.
Figure \ref{fig:traineplength} shows the restricted episode length values applied as constraints during the training phase.
With these restrictions applied, the number of samples in interaction steps was encouraged to be concentrated in the early part of the episode, as shown in Figure \ref{fig:visit}.
Due to this intensive data collection early in the episode, AELA can reduce visitation to unnecessary states later in the episode, particularly dead-end states.
We analyzed cases in SMAC where the presence of dead-end states prevents finding a winning policy. 
Figure \ref{fig:testeplength} shows the episode length during testing. 
There is a tendency for QMIX to converge to a policy that has longer episode lengths, which is associated with dead-end states.

To examine this situation in more detail, a snapshot of the actual final policy is shown in Figure \ref{fig:snapaela} and \ref{fig:snapqmix}.
As shown in Figures \ref{fig:snapaela}, in the 	\textit{6h\_vs\_8z} scenario, an effective strategy involves having a few allied units lure some enemy units away. 
This prevents those enemy units from participating in the battle, allowing the remaining allied units to first attack the rest of the enemy units, and then finally engage the lured enemy units.
However, as shown in Figure \ref{fig:snapqmix}, if the strategy of luring enemy units is not properly learned, it can lead to an incorrect strategy where an allied unit moves away alone from the rest of the allied units, even when the enemy units are not following.
Such a state, where all allied units die in combat against enemy units, leaving one isolated allied unit due to the failure of luring, can be considered a dead-end state. This is because a single allied unit cannot defeat all the remaining enemy units, making it impossible to transition to a winning state.
Therefore, as shown in Figure \ref{fig:snapqmix}, all states after $l=38$ can be considered dead-end states.
AELA mitigates the risk of entering dead-end states by terminating episodes at appropriate times and starting new episodes, which results in improved performance.

\section{Conclusion}
In this paper, we presented an adaptive approach to episode length adjustment for MARL. 
We theoretically demonstrated that limiting the episode length reduces the visitation probability of dead-end states, thereby allowing the learned value function to be closer to the optimal value function.
Based on this background, we proposed a practical algorithm that initially limits the episode length and gradually extends it while assessing convergence through the entropy of Q-values.
The proposed AELA algorithm was experimentally verified to outperform the traditional fixed episode length approach across various scenarios in multi-agent environments.

To the best of our knowledge, this paper is the first attempt to improve learning performance in multi-agent reinforcement learning by adjusting the episode length.
Adjusting the episode length can be applied to any method, regardless of value decomposition techniques, additional intrinsic rewards, or network architecture, suggesting the potential for performance improvement from a new perspective.
Future work could focus on extending this approach to more complex domains and incorporating additional metrics that could further enhance the adaptiveness of the episode length, ensuring optimal policy convergence even in highly dynamic environments.


\begin{acks}
This work was supported in part by Electronics and Telecommunications Research Institute (ETRI) grants funded by the Korean Government [25ZB1100, Core Technology Research for Self-Improving Integrated Artificial Intelligence System] and by the ETRI Research and Development Support Program of MSIT/IITP [RS-2023-00216821, Development of Beyond X-verse Core Technology for Hyper-realistic Interactions by Synchronizing the Real World and Virtual  Virtual Space].
\end{acks}



\bibliographystyle{ACM-Reference-Format} 

\balance
\bibliography{BYOO_AAMAS25}

\clearpage
\input{appendix}


\end{document}

%% file: appendix.tex
\appendix
\onecolumn
\section{Proof of Theorem \ref{theorem:increase}}
\label{appendix:proof_increase}

\renewcommand{\thetheorem}{1}
\begin{theorem}
Let $E_L$ be the fixed episode length, and $N_{s}$ be the expected number of secure state visits over all episodes.
Then, $N_{s}$ remains constant or increases if $E_L$ is reduced.
\label{theorem:increase}
\end{theorem}

\begin{proof}[\textbf{Proof of Theorem \ref{theorem:increase}}]
This is the proof process for Theorem \ref{theorem:increase}.
$N_{\text{total}}$ is the total number of collected data samples, 
$E_L$ is the episode length (i.e., the number of time steps in an episode), and $E_N$ is the number of episodes in the total data samples. 
$P_{s}(l)$ is the probability that the state is secure at the $l$-th interaction step of an episode, $E_{s}(l)$ is the expected number of secure state visits per episode, and $N_{s}$ is the expected number of secure state visits across all episodes.
Based on the above definitions, this equation holds: 
\begin{align}
E_L \times E_N = N_{\text{total}} \implies E_N = \frac{N_{\text{total}}}{E_L}.
\end{align}
The expected number of secure state visits per episode is given by:
\begin{align}
E_{s} = \sum_{l=1}^{E_L} P_{s}(l).
\end{align}
Then, the expected number of secure state visits across all episodes, $N_{s}$, is given by:
\begin{align}
N_{s} = E_N \times E_{s} = \frac{N_{\text{total}}}{E_L} \times \left( \sum_{l=1}^{E_L} P_{s}(l) \right).
\end{align}
Considering \( N_{s} \) as a function of \( E_L \), we examine how \( N_{s} \) changes as \( E_L \) increases.
\begin{align}
N_{s}(E_L) = \frac{N_{\text{total}}}{E_L} \times \left( \sum_{l=1}^{E_L} P_{s}(l) \right).
\end{align}
Our goal is to show that \( N_{s}(E_L) \) increases or remains constant as \( E_L \) increases. To do this, we calculate the change in \( N_{s}(E_L) \), denoted as \( \Delta N_{s} \).
Since the episode length is a discrete value, the change in \( N_{s} \) when \( E_L \) increases from \( E_L \) to \( E_L + 1 \) is given by:
\begin{align}
\Delta N_{s} = N_{s}(E_L + 1) - N_{s}(E_L).
\end{align}
Substituting the values:
\begin{align}
N_{s}(E_L + 1) = \frac{N_{\text{total}}}{E_L + 1} \times \left( \sum_{l=1}^{E_L + 1} P_{s}(l) \right).
\end{align}
\begin{align}
N_{s}(E_L) = \frac{N_{\text{total}}}{E_L} \times \left( \sum_{l=1}^{E_L} P_{s}(l) \right).
\end{align}
Thus, the change is:
\begin{align}
\Delta N_{s} = N_{\text{total}} \times \left[ \frac{1}{E_L + 1} \left( \sum_{l=1}^{E_L + 1} P_{s}(l) \right) - \frac{1}{E_L} \left( \sum_{l=1}^{E_L} P_{s}(l) \right) \right].
\end{align}
To determine the sign of \( \Delta N_{s} \), we simplify the expression. We combine the terms by finding a common denominator:
\begin{align}
\Delta N_{s} = N_{\text{total}} \times \left[ \frac{E_L \left( \sum_{l=1}^{E_L + 1} P_{s}(l) \right) - (E_L + 1) \left( \sum_{l=1}^{E_L} P_{s}(l) \right)}{E_L (E_L + 1)} \right].
\end{align}
Expanding and simplifying the numerator:
\begin{align}
\Delta N_{s} &= N_{\text{total}} \times \left[ \frac{E_L \sum_{l=1}^{E_L + 1} P_{s}(l) - (E_L + 1) \sum_{l=1}^{E_L} P_{s}(l)}{E_L (E_L + 1)} \right] \\
&= N_{\text{total}} \times \left[ \frac{E_L \left( \sum_{l=1}^{E_L} P_{s}(l) + P_{s}(E_L + 1) \right) - (E_L + 1) \sum_{l=1}^{E_L} P_{s}(l)}{E_L (E_L + 1)} \right] \\
&= N_{\text{total}} \times \left[ \frac{E_L P_{s}(E_L + 1) - \sum_{l=1}^{E_L} P_{s}(l)}{E_L (E_L + 1)} \right].
\end{align}
Thus, $\Delta N_{s} = \frac{N_{\text{total}}}{E_L (E_L + 1)} \left( E_L P_{s}(E_L + 1) - \sum_{l=1}^{E_L} P_{s}(l) \right)$.
The sign of $\Delta N_{s}$ depends on the sign of the numerator.
Since $P_{s}(l)$ is a monotonically non-increasing function, the following holds for all $l = 1, 2, \dotsc, E_L$:
\begin{align}
P_{s}(l) \geq P_{s}(E_L + 1).
\end{align}
Using the above property, we derive an inequality for the sum.
\begin{align}
\sum_{l=1}^{E_L} P_{s}(l) \geq \sum_{l=1}^{E_L} P_{s}(E_L + 1) = E_L P_{s}(E_L + 1).
\end{align}
Thus, the following holds:
\begin{align}
\sum_{l=1}^{E_L} P_{s}(l) \geq E_L P_{s}(E_L + 1).
\end{align}
Using this inequality, the value of the numerator becomes:
\begin{align}
E_L P_{s}(E_L + 1) - \sum_{l=1}^{E_L} P_{s}(l) \leq 0.
\end{align}
Thus, the numerator is less than or equal to 0, which implies that \( \Delta N_{s} \leq 0 \).

Since \( \Delta N_{s} \leq 0 \), \( N_{s}(E_L) \) is non-increasing as \( E_L \) increases. This means that decreasing the episode length \( E_L \) will either keep $N_{s}$ constant or cause it to increase.
\end{proof}

\section{Proof of Theorem \ref{theorem:regret}}
\label{appendix:proof_regret}

\renewcommand{\thetheorem}{2}
\begin{theorem}
Under Assumption \ref{assumption:goal_reward}, as \( P_d \) decreases, the difference in cumulative rewards between the optimal policy and the agent's policy also decreases, thereby reducing Regret(T) (defined in Definition \ref{definition:regret}).
\label{theorem:regret}
\end{theorem}

\begin{proof}[\textbf{Proof of Theorem \ref{theorem:regret}}]
We will expand the regret after $T$ time steps.
The probability of visiting dead-end states and goal states are denoted as $P_d$ and $P_g$, respectively.
Here, we assume that $P_d$ and $P_g$ are constant.
The regret function is given as follows:
\begin{align}
\text{Regret}(T) &= R(\pi^*) - R(\pi) \\
&= \left( \sum_{t=1}^T r_t^* + r_g \right) - \left( \sum_{k=1}^T (1 - P_d)^{k-1} P_g \left( \sum_{t=1}^k r_t + r_g \right) - \left( 1 - \sum_{k=1}^T (1 - P_d)^{k-1} P_g \right) \sum_{t=1}^T r_t \right).
\label{eq:regret1}
\end{align}
\color{magenta}

\color{black}

Equation \ref{eq:regret1} is grouped into three major brackets. The first group represents $R(\pi^*)$, while the second and third groups refer to $R(\pi)$. 
The second group denotes the expected reward received at each time step from $t=1$ to $T$ when the goal state is reached, and the third group represents the sum of rewards when the goal state is not reached by time $T$.
Equation \ref{eq:regret1} is expanded as follows:
\begin{align}
&= \sum_{t=1}^T r_t^* + r_g - \sum_{k=1}^T (1 - P_d)^{k-1} P_g \sum_{t=1}^k r_t - r_g \sum_{k=1}^T (1 - P_d)^{k-1} P_g - \sum_{t=1}^T r_t + \sum_{k=1}^T (1 - P_d)^{k-1} P_g \sum_{t=1}^T r_t \\
&= \sum_{t=1}^T (r_t^* - r_t) + \sum_{k=1}^T (1 - P_d)^{k-1} P_g \left( \sum_{t=1}^T r_t - \sum_{t=1}^k r_t - r_g \right) + r_g.
\end{align}
To summarize, the regret is as follows:
\begin{align}
\text{Regret}(T) = \sum_{t=1}^T (r_t^* - r_t) + \sum_{k=1}^T (1 - P_d)^{k-1} P_g \left( \sum_{t=1}^T r_t - \sum_{t=1}^k r_t - r_g \right) + r_g.
\end{align}
Now, we differentiate the regret with respect to \( P_d \). The term that is directly affected by \( P_d \) is the second term:
\begin{align}
\frac{d}{d P_d} \text{Regret}(T) = \frac{d}{d P_d} \sum_{k=1}^T (1 - P_d)^{k-1} P_g \left( \sum_{t=1}^T r_t - \sum_{t=1}^k r_t - r_g \right).
\end{align}
First, we differentiate \( (1 - P_d)^{k-1} \):
\begin{align}
\frac{d}{d P_d} (1 - P_d)^{k-1} = -(k-1)(1 - P_d)^{k-2}.
\end{align}
Applying this result to the regret function:
\begin{align}
\frac{d}{d P_d} \text{Regret}(T) = \sum_{k=1}^T -(k-1)(1 - P_d)^{k-2} P_g \left( \sum_{t=1}^T r_t - \sum_{t=1}^k r_t - r_g \right).
\label{eq:regret_derivative}
\end{align}
The term \( -(k-1) \) is always negative except when $k=1$.
The term \( (1 - P_d)^{k-2} \) is always positive.
The term \( P_g \) is also positive.
If the term is \( \left( \sum_{t=1}^T r_t - \sum_{t=1}^k r_t - r_g \right) \) is negative, the derivative of regret in Equation \ref{eq:regret_derivative} is positive.
Now, considering the expression \( \left( \sum_{t=1}^T r_t - \sum_{t=1}^k r_t - r_g \right) \):

\begin{align}
\sum_{t=1}^T r_t - \sum_{t=1}^k r_t - r_g = \sum_{t=k+1}^T r_t - r_g.
\end{align}
Based on Assumption \ref{assumption:goal_reward}, this term is negative.
Therefore, this inequality holds:
\begin{align}
\frac{d}{d P_d} \text{Regret}(T) > 0.
\end{align}
Thus, under the assumption \ref{assumption:goal_reward} reducing \( P_d \) decreases the regret, as desired.
\end{proof}

\section{Hyper-parameters}
Essentially, we used the parameters reported in the original papers of VDN and QMIX without any modifications to ensure fairness in the experiments.
All algorithms employ a utility function that integrates a GRU between two fully connected (FC) layers containing $64$ neurons, with ReLU activations. 
QMIX relies on a non-linear mixing network, where hypernetworks with either one or two layers of $64$ neurons and ELU activations generate the parameters. 
The size of the replay buffer is set to $5000$ and the target network is updated in every $200$ steps.
$\epsilon$ decays from $1.0$ to $0.05$ over $50000$ time steps for easy scenarios in SMAC (\textit{3m} and \textit{2s\_vs\_1sc}), $1$ million time steps in difficult scenarios in SMAC (\textit{MMM2}, \textit{3s5z\_vs\_3s6z}, \textit{6h\_vs\_8z}, \textit{Corridor}), and $500000$ time steps in MPP tasks.
For AELA, two additional hyper-parameters are used: the initial episode length $E_{L_{0}}$ and the window size $w$.
$E_{L_{0}}$ is set to $25\%$ of the original episode length of each scenario ($E_{max}/4$).
We chose the better value for the window size $w$ between 150 and 300 in all scenarios except \textit{6h\_vs\_8z} in SMAC.
As described in the main text, for the \textit{6h\_vs\_8z} scenario, where dead-end states frequently occur, we focused on exploring the initial states of episodes by setting $E_{L_{0}}$ to $15\%$ of the original episode length and setting the window size to 900.

If the episode length increases too quickly relative to the training time, the difference from the traditional method of fixing the episode length at the maximum becomes negligible. Therefore, when setting $w$, we calculated the time step at which the episode length reaches its maximum value when increased at the fastest possible rate and used this calculation to determine $w$.
AELA, like QMIX, collects one episode and updates the $q$-function once. 
During each update, one entropy value is calculated. 
Once $w$ entropy values are collected, a decision is made about whether to increase the episode length. 
When the average episode length is $\bar{L}$, the decision of whether to increase the episode length is made, on average, every $\bar{L} \cdot w$ time steps. 
Thus, during $t$ time steps, the maximum increased number of episode lengths to increase is $\frac{t}{\bar{L} \cdot w}$. 
Let the difference between the initial episode length limit $E_{L_{0}}$ and the maximum episode length $E_{max}$ be $\Delta_{L} = E_{max} - E_{L_{0}}$.
We aim to prevent the episode length from reaching its maximum value before a certain training time $t$.
In other words, the maximum possible increase in episode length during $t$ time steps should be less than or equal to $\Delta_{L}$.
This condition can be written as:
\begin{align}
\frac{t}{\bar{L} \cdot w} \leq \Delta_{L}.
\end{align}
Rearranging for $w$:
\begin{align}
w \geq \frac{t}{\bar{L} \cdot \Delta_{L}}.
\end{align}
Since the exact value of $\bar{L}$ is unknown, it is approximated using the average of the initial episode length limit and the maximum episode length:
\begin{align}
\bar{L} \approx \frac{E_{max} + E_{L_{0}}}{2}.
\end{align}
Using this formula, it is possible to set $w$ to ensure that the episode length does not reach its maximum too quickly during the training period. 
In this study, $t$ was set to approximately 80\% of the total training time, and the value of $w=300$ used in most scenarios in this paper was determined using this method.

\color{black}
